\documentclass[conference]{IEEEtran}
\usepackage{}
\usepackage{mathrsfs,amsmath,amssymb,epsfig,amsfonts,fancyhdr,graphicx,cite,amsthm}

\newtheorem{thm}{Theorem}[]
\newtheorem{cor}{Corollary}[]

\theoremstyle{remark}
\newtheorem{rem}[]{Remark}

\theoremstyle{definition}

\begin{document}

\title{Compress-Forward without Wyner-Ziv Binning for the One-Way and Two-Way Relay Channels}
\author{\authorblockN{Peng Zhong and Mai Vu\\}
\authorblockA{Department of Electrical and Computer Engineering\\
McGill University\\
Montreal, QC, Canada H3A 2A7\\
Emails: peng.zhong@mail.mcgill.ca, mai.h.vu@mcgill.ca}}

\maketitle

\begin{abstract}
We consider the role of Wyner-Ziv binning in compress-forward for relay
channels. In the one-way relay channel, we analyze a compress-forward scheme
without Wyner-Ziv binning but with joint decoding of both the message and compression index.
It achieves the same rate as the original compress-forward scheme with binning and successive decoding.
Therefore, binning helps reduce decoding complexity by allowing successive decoding, but
has no impact on achievable rate for the one-way relay channel. On the other hand, no binning simplifies relay operation. By extending compress-forward without binning to the two-way relay channel, we can achieve a larger rate region than the original compress-forward
scheme when the channel is asymmetric for the two users. Binning and successive decoding limits the compression rate to match the weaker of the channels from relay to two users, whereas without binning, this restriction no longer applies. Compared with noisy network coding, compress-forward
 without binning
achieves the same rate region in certain Gaussian channel configurations, and it has much less delay. This work is a step toward understanding the role of Wyner-Ziv binning in compress-forward relaying.
\end{abstract}

\IEEEpeerreviewmaketitle
\section{Introduction}\label{sec:intro}
\IEEEPARstart{T}{he} relay channel (RC) first introduced by van der Meulen \cite{meulen1971three} is a 3-node channel, in which a sender aims to communicate with a receiver with the help of a relay. Several coding schemes for the discrete-memoryless relay channel have been established. Compress-forward is a scheme proposed by Cover and El Gamal in \cite{cover1979capacity}, in which the relay compresses its noisy observation of the source signal and forwards the bin index of the compression to the receiver using Wyner-Ziv coding \cite{Wyner1976the}. Successive decoding is then performed at the receiver. At the end of each block, the receiver decodes the compression index first, then uses that to decode the message sent in the previous block. In \cite{gamal2006bounds}, El Gamal, Mohseni, and Zahedi put forward an equivalent form of the compress-forward lower bound. In \cite{rankov2006achievable}, Rankov and Wittneben apply the compress-forward scheme to the two-way relay channel (TWRC) in which two users wish to exchange messages with the help of a relay.\par

Recently, Lim, Kim, El Gamal and Chung put forward a noisy network coding scheme \cite{sung2011noisy} for the general multi-source noisy network. The scheme involves lossy compression by the nodes as in compress-forward for the relay channel. However, unlike compress-forward in which independent messages are sent over multiple blocks, here the same message is sent multiple times using independent codebooks. Furthermore, the nodes use no Wyner-Ziv binning, and perform simultaneous decoding of the received signals from all blocks without uniquely decoding the compression indices. Noisy network coding simplifies to the capacity-achieving network coding for noiseless multicast networks and achieves a larger rate region than the original compress-forward when applied to multisource networks such as the two-way relay channel. However, it also brings more delay in decoding.\par

Motivated by the original compress-forward and the new noisy network coding schemes, we aim to understand the role of binning by analyzing a compress-forward scheme in which the relay does not use Wyner-Ziv binning, and the receiver performs joint decoding of both the message and compression index based on signals received from both the current and previous blocks. In the one-way relay channel, compress-forward without binning achieves the same rate as the original compress-forward scheme and noisy network coding. Comparing with the original compress-forward, it simplifies relay operation since Wyner-Ziv binning is not needed, but increases decoding complexity since joint decoding instead of successive decoding is required. We then extend compress-forward without binning to the two-way relay channel. We show that it achieves strictly larger rate region than the original compress-forward scheme as in \cite{rankov2006achievable} when the channel is asymmetric for the two users. Although it
generally achieves smaller rate region than noisy network coding, it only has one block decoding delay. For the Gaussian TWRC, we also provide specific conditions for when compress-forward without binning achieves the same rate region as noisy network coding.\par
The remainder of this paper is organized as follows. We present the
channel models in Section \ref{sec:system_model}. The compress-forward
scheme without binning is applied to the one-way relay channel in Section \ref{sec:one-way}. In
Section \ref{sec:two-way}, we extend it to the two-way relay channel and present
numerical results. Finally, we conclude the paper
in Section \ref{Con}.


\section{Channel Models}\label{sec:system_model}
\subsection{Discrete Memoryless RC Model}
The discrete memoryless one-way relay channel (DM-RC) is
denoted by $(\mathcal{X} \times \mathcal{X}_r ,p(y,y_r|x,x_r),\mathcal{Y} \times
\mathcal{Y}_r)$, as in Figure \ref{fig:one-way}. Sender $X$ wishes to send a message $M$ to receiver $Y$ with the help of the relay $(X_r,Y_r)$. We consider a full-duplex
channel in which all nodes can transmit and receive at the same time.

\begin{figure}[t]
\centering
  \includegraphics[scale=0.7]{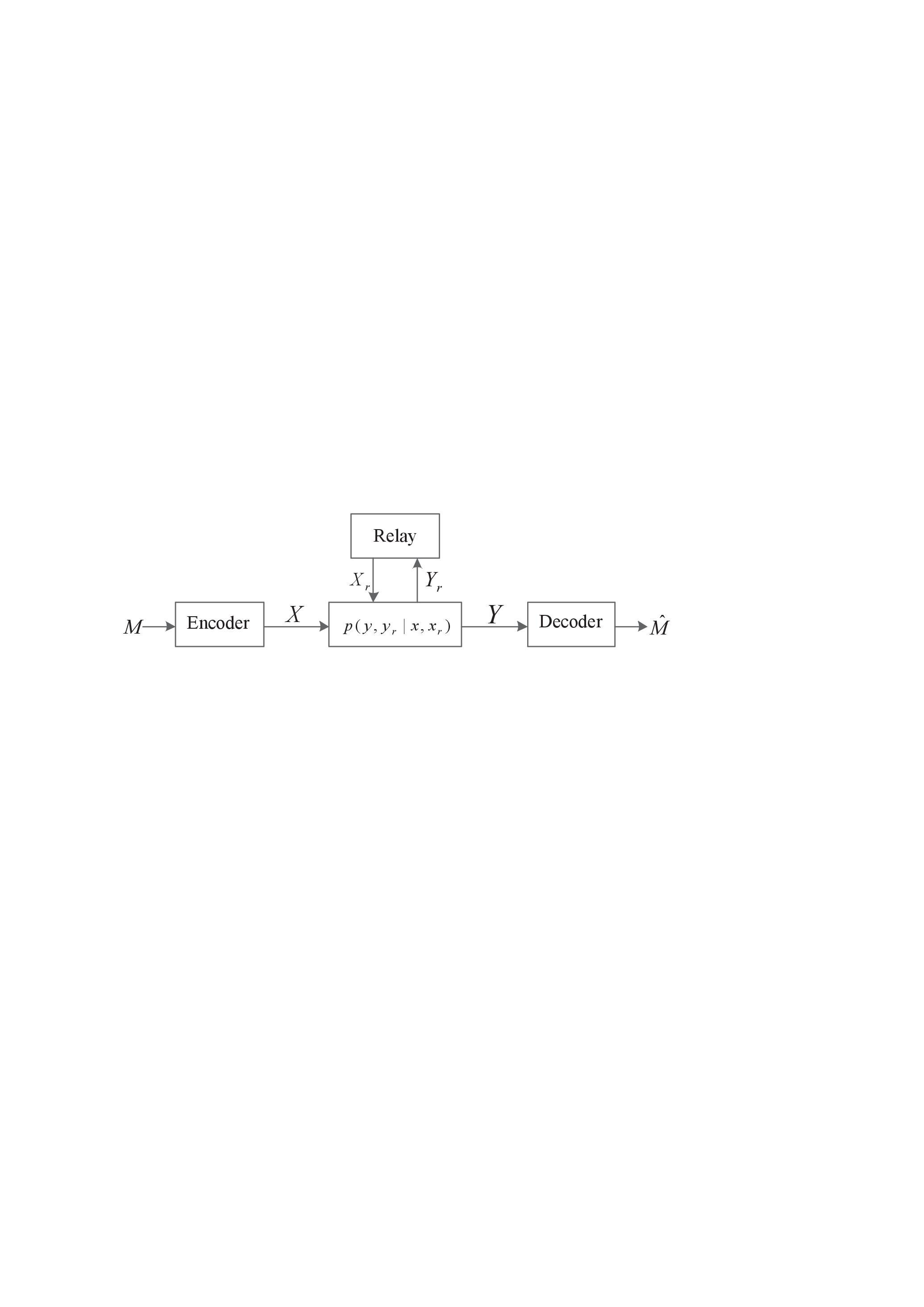}\\
  \caption{One-way relay channel model}\label{fig:one-way}
\end{figure}

A $(n,2^{nR},P_e)$ code for a DM-RC consists of: a
message set $\mathcal{M}=[1:2^{nR}]$; an encoder that assigns a codeword $x^n(m)$ to each message $m\in [1:2^{nR}]$; a relay encoder that assigns at time $i\in [1:n]$ a symbol $x_{ri}(y_r^{i-1})$ to each past received sequence $y_r^{i-1}\in \mathcal{Y}_r^{i-1}$; a decoder that assigns a message $\hat{m}$ or an error message to each received sequence $y^{n}\in \mathcal{Y}^{n}$. The average error probability is
$P_e=\textrm{Pr}\{\hat{M}\neq M\}$. The rate $R$ is said to be achievable for the DM-RC if there exists a sequence of $(2^{nR},n)$ codes with $P_e\rightarrow 0$. The supremum of all achievable rates is the capacity of the DM-RC.

\subsection{Discrete Memoryless TWRC Model}
The discrete memoryless two-way relay channel (DM-TWRC) is
denoted by $(\mathcal{X}_1 \times \mathcal{X}_2 \times
\mathcal{X}_r,p(y_1,y_2,y_r|x_1,x_2,x_r),\mathcal{Y}_1 \times
\mathcal{Y}_2 \times \mathcal{Y}_r)$, as in Figure \ref{p2}. Here
$x_1$ and $y_1$ are the input and output signals of user 1; $x_2$ and
$y_2$ are the input and output signals of user 2; $x_r$ and $y_r$ are
the input and output signals of the relay. We also consider a full-duplex
channel in which all nodes can transmit and receive at the same time.

\begin{figure}[t]
\centering
  \includegraphics[scale=0.7]{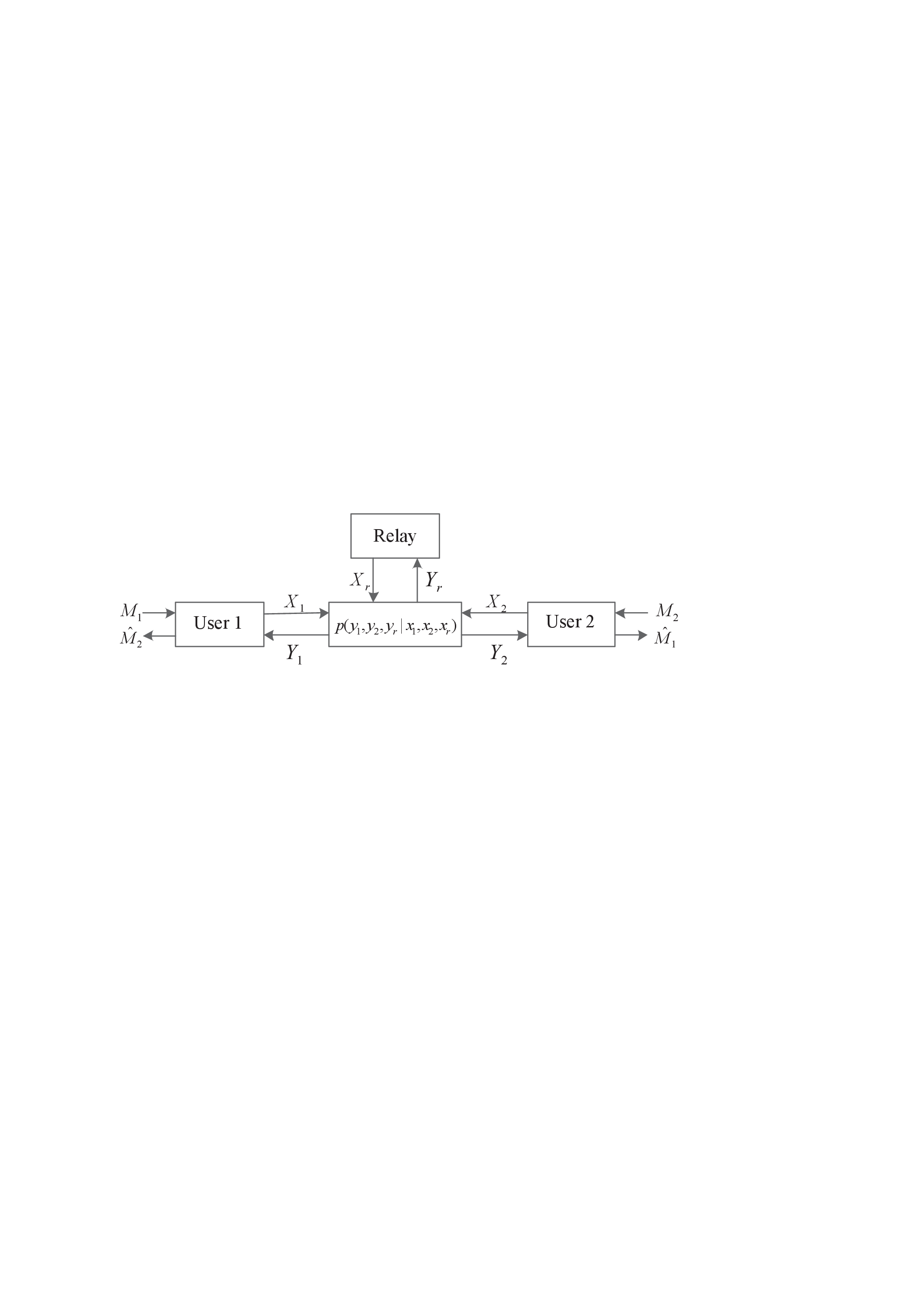}\\
  \caption{Two-way relay channel model}\label{p2}
\end{figure}

A $(n,2^{nR_1},2^{nR_2},P_e)$ code for a DM-TWRC consists of two
message sets $\mathcal{M}_1=[1:2^{nR_1}]$ and
$\mathcal{M}_2=[1:2^{nR_2}]$, three encoding functions
$f_{1,i},f_{2,i},f_{r,i}$, $i=1, \ldots, n$ and two decoding function
$g_1,g_2$.
\begin{align}
x_{1,i}&=f_{1,i}(M_1,Y_{1,1},\ldots ,Y_{1,i-1}),~~~~i=1, \ldots, n\nonumber\\
x_{2,i}&=f_{2,i}(M_2,Y_{2,1},\ldots ,Y_{2,i-1}),~~~~i=1, \ldots, n\nonumber\\
x_{r,i}&=f_{r,i}(Y_{r,1},\ldots ,Y_{r,i-1}),~~~~i=1, \ldots, n\nonumber\\
g_1&:\mathcal{Y}^n_1 \times \mathcal{M}_1 \rightarrow \mathcal{M}_2\nonumber\\
g_2&:\mathcal{Y}^n_2 \times \mathcal{M}_2 \rightarrow \mathcal{M}_1\nonumber
\end{align}
The average error probability is
$P_e=\textrm{Pr}\{g_1(M_1,Y^n_1)\neq
M_2~\textrm{or}~g_2(M_2,Y^n_2)\neq M_1\}$. A rate pair is said to be
achievable if there exists a $(n,2^{nR_1},2^{nR_2},P_e)$ code such
that $P_e\rightarrow 0$ as $n\rightarrow \infty$. The closure of the
set of all achievable rates $(R_1,R_2)$ is the capacity region of the
two-way relay channel.

\subsection{Gaussian TWRC Model}

\begin{figure}[t]
\centering
  \includegraphics[scale=0.6]{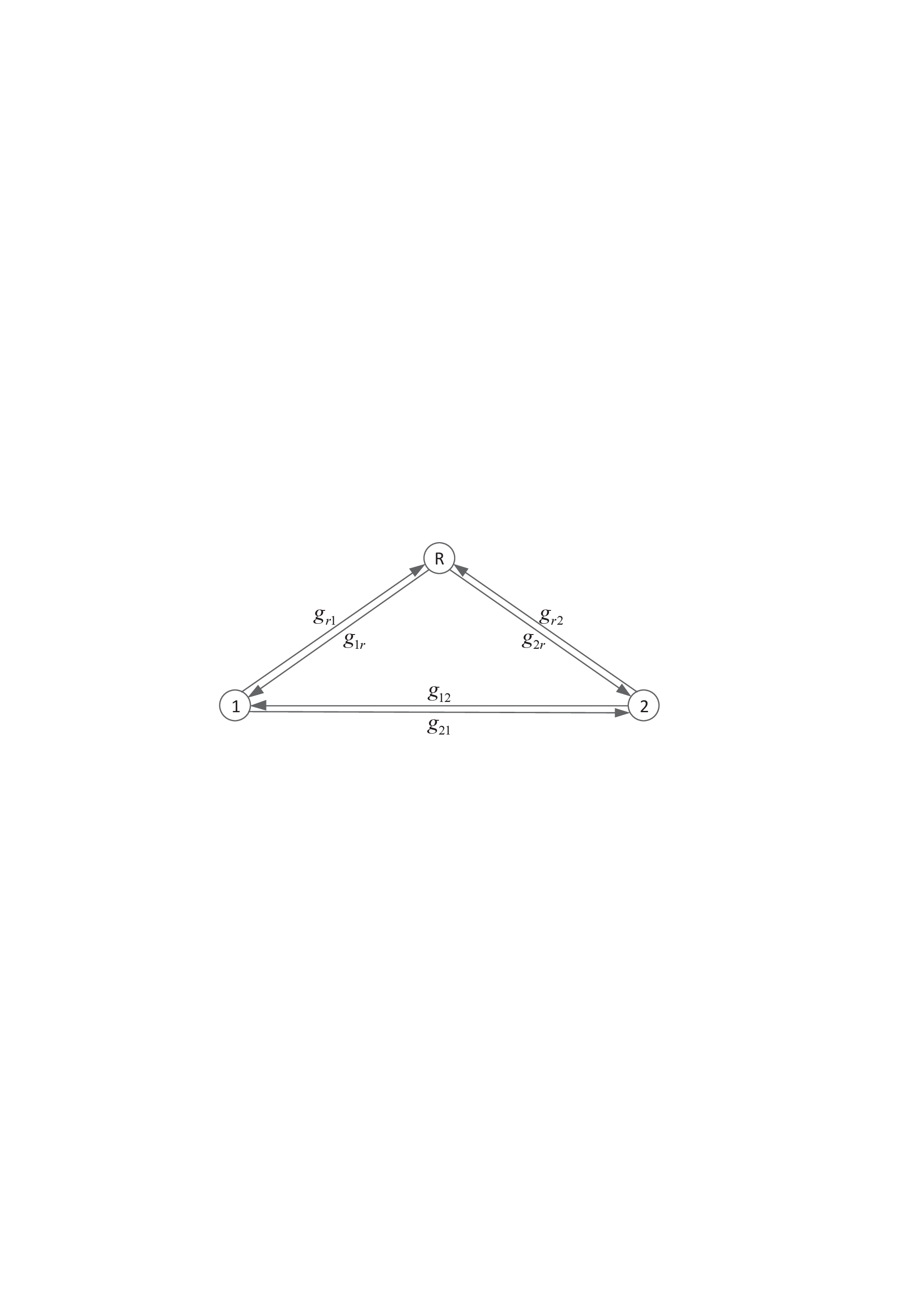}\\
  \caption{Gaussian two-way relay channel model}\label{fig:GTWRC}
\end{figure}

As in Figure \ref{fig:GTWRC}, the Gaussian two-way relay channel can be modeled as:
\begin{align}
  Y_1&=g_{12}X_2+g_{1r}X_r+Z_1 \nonumber \\
  Y_2&=g_{21}X_1+g_{2r}X_r+Z_2 \nonumber \\
  Y_r&=g_{r1}X_1+g_{r2}X_2+Z_r
  \label{GTWRC}
\end{align}
where $Z_1, Z_2, Z_r\sim\mathcal{N}(0,1)$ are independent Gaussian noises and $g_{12}, g_{1r}, g_{21}, g_{2r}, g_{r1}, g_{r2}$ are corresponding channel gains. The average
input power constraints for user 1, user 2 and the relay are all
$P$.

\section{One-Way Relay Channel}\label{sec:one-way}
In the original compress-forward scheme \cite{cover1979capacity} \cite{gamal2006bounds} , the relay forwards the bin index of the description of its received signal. The receiver uses successive decoding to decode the bin index first and then decode the message from the sender. Here we analyze a compress-forward scheme in which the relay forwards the description index directly while the receiver jointly decodes the index and the message at the same time. We show that compress-forward without binning can achieve the same rate as the original compress-forward scheme with binning. Note that different from noisy network coding \cite{sung2011noisy}, compress-forward without binning sends a different message at each block without message repetition.

\subsection{Achievable Rate for Compress-Forward without Binning}
\begin{thm}
\label{thm_rate}
Consider a compress-forward scheme in which the relay does not use Wyner-Ziv binning but sends the compression index directly and the receiver performs joint decoding of both the message and compression index.
The following rate is achievable for one-way relay channel:
\begin{align}
R\leq \min \{&I(X,X_r;Y)-I(\hat{Y}_r;Y_r|X,X_r,Y),\nonumber\\
&I(X;Y,\hat{Y}_r|X_r)\}
\end{align}
subject to
\begin{align}
\label{oneway_const_1}
I(X_r;Y)+I(\hat{Y}_r;X,Y|X_r)\geq I(\hat{Y}_r;Y_r|X_r)
\end{align}
for some $p(x)p(x_r)p(\hat{y}_r|y_r,x_r)p(y,y_r|x,x_r)$.
\end{thm}

\begin{proof}
We use a block coding scheme in which each user sends $b-1$ messages over $b$ blocks of $n$ symbols each.\par
\subsubsection{Codebook generation}
Fix $p(x)p(x_r)p(\hat{y}_r|y_r,x_r)$. We randomly and independently generate a codebook for each block $j\in [1:b]$
\begin{itemize}
\item Generate $2^{nR}$ i.i.d. sequences $x^n(m_j)\sim\prod ^n_{i=1}p(x_i) $, where $m_j \in [1:2^{nR}]$.
\item Generate $2^{nR_r}$ i.i.d. sequences $x_r^n(k_{j-1})\sim\prod ^n_{i=1}p(x_{ri})$, where $k_{j-1} \in [1:2^{nR_r}]$.
\item For each $k_{j-1}\in [1:2^{nR_r}]$, generate $2^{nR_r}$ i.i.d. sequences $\hat{y}^n_r(k_j|k_{j-1})\sim\prod ^n_{i=1}p(\hat{y}_{ri}|x_{ri}(k_{j-1})$, where $k_{j}\in [1:2^{nR_r}]$.
\end{itemize}

\subsubsection{Encoding}
The sender transmits $x^n(m_j)$ in block $j$. The relay, upon receiving $y^n_r(j)$, finds an index $k_j$ such that $
((\hat{y}_r^n(k_j|k_{j-1}),y^n_r(j),x^n_r(k_{j-1}))\in A^n_{\epsilon '} $. Assume that such $k_j$ is found, the relay sends $x^n_r(k_j)$ in block $j+1$.\par

\subsubsection{Decoding}
 Assume the receiver has decoded $k_{j-1}$ correctly in block $j$. Then in block $j+1$, the receiver finds a unique pair of $(\hat{m}_j,\hat{k}_j)$ such that
\begin{align}
(x_r^n(\hat{k}_j),y^n(j+1))&\in A^n_{\epsilon} \nonumber \\
\textrm{and}~~~~(x^n(\hat{m}_j),x_r^n(\hat{k}_{j-1}),\hat{y}_r^n(\hat{k}_j|\hat{k}_{j-1}),y^n(j))&\in A^n_{\epsilon}. \nonumber
\end{align}

\subsubsection{Error analysis}
 Assume without loss of generality that $m_j=1$ and $k_{j-1}=k_{j}=1$. First define the following two events:
\begin{align}
\mathcal{E}'_{1j}(k_j)&=\big\{(x_r^n(k_j),y^n(j+1))\in A^n_{\epsilon}\big\} \nonumber \\
\mathcal{E}'_{2j}(m_j,k_j)&=\big\{(x^n({m}_j),x_r^n(1),\hat{y}_r^n({k}_j|1),y^n(j))\in A^n_{\epsilon}\big\}. \nonumber
\end{align}
Then the decoder makes an error only if one or more of the following events occur:
\begin{align}
\mathcal{E}_{1j}=&\big\{(\hat{y}_r^n(k_j|1),y^n_r(j),x^n_r(1))\notin A^n_{\epsilon '}~\textrm{for all}~k_j\in[1:2^{nR_r}]\big\} \nonumber \\
\mathcal{E}_{2j}=&\big\{(x_r^n(1),y^n(j+1))\notin A^n_{\epsilon}~\textrm{or}~(x^n(1),x_r^n(1),\hat{y}_r^n(1|1),\nonumber \\&y^n(j))\notin A^n_{\epsilon} \big\} \nonumber\\
\mathcal{E}_{3j}=&\big\{\mathcal{E}'_{1j}(k_j)~\textrm{and}~\mathcal{E}'_{2j}(1,k_j)~\textrm{for some}~k_j\neq 1\big\} \nonumber \\
\mathcal{E}_{4j}=&\big\{\mathcal{E}'_{1j}(1)~\textrm{and}~\mathcal{E}'_{2j}(m_j,1)~\textrm{for some}~m_j\neq 1\big\} \nonumber \\
\mathcal{E}_{5j}=&\big\{\mathcal{E}'_{1j}(k_j)~\textrm{and}~\mathcal{E}'_{2j}(m_j,k_j)~\textrm{for some}~m_j\neq 1,k_j\neq 1\big\}. \nonumber
\end{align}
Thus, the probability of error is bound as
\begin{align}
P\{\hat{m_j}\neq 1,\hat{k_j}\neq 1\}\leq & P(\mathcal{E}_{1j})+P(\mathcal{E}_{2j}\cap\mathcal{E}_{1j}^c)+P(\mathcal{E}_{3j})\nonumber \\&+P(\mathcal{E}_{4j})+P(\mathcal{E}_{5j}). \nonumber
\end{align}
By the covering lemma\cite{gamal2010lecture}, $P(\mathcal{E}_{1j})\rightarrow 0$ as $n\rightarrow \infty$, if
\begin{align}
\label{rate_oneway_1}
R_r > I(\hat{Y}_r;Y_r|X_r).
\end{align}
By the conditional typicality lemma \cite{cover2006elements}, $P(\mathcal{E}_{2j}\cap\mathcal{E}_{1j}^c) \rightarrow 0$  as $n\rightarrow \infty$.\par
For the rest of the error events, the decoded joint distribution for each event is as follows.
\begin{align*}
\mathcal{E}'_{1j}(k_j)&: p(x_r)p(y)\\
\mathcal{E}'_{2j}(1,k_j)&: p(x)p(x_r)p(\hat{y}_r|x_r)p(y|x,x_r)\\
\mathcal{E}'_{2j}(m_{j},1)&: p(x)p(x_r)p(y,\hat{y}_r|x_r)\\
\mathcal{E}'_{2j}(m_j,k_j)&: p(x)p(x_r)p(\hat{y}_r|x_r)p(y|x_r),
\end{align*}
where $m_j \neq 1, k_j \neq 1$. Using standard joint typicality analysis with the above decoded joint distribution, we can obtain a bound on each error event as follows.
\begin{align}
P(\mathcal{E}_{3j}) &\leq 2^{nR_r}\cdot 2^{-n(I(X_r;Y)-\delta(\epsilon))} \cdot 2^{-n(I(\hat{Y}_r;X,Y|X_r)-3\delta(\epsilon))} \nonumber\\
P(\mathcal{E}_{4j}) &\leq 2^{nR}\cdot 2^{-n(I(X;Y,\hat{Y}_r|X_r)-2\delta(\epsilon))} \nonumber\\
P(\mathcal{E}_{5j}) &\leq2^{nR}\cdot 2^{nR_r}\cdot 2^{-n(I(X_r;Y)-\delta(\epsilon))}\nonumber \\
&~~~\cdot 2^{-n(I(X;Y|X_r)+I(\hat{Y}_r;X,Y|X_r)-3\delta(\epsilon))}. \nonumber
\end{align}
All of them tend to zero as $n\rightarrow \infty $ if
\begin{align}
\label{rate_oneway_2}
R_r &\leq I(X_r;Y)+I(\hat{Y}_r;X,Y|X_r)\\
\label{rate_oneway_3}
R &\leq I(X;Y,\hat{Y}_r|X_r)\\
\label{rate_oneway_4}
R+R_r &\leq I(X_r;Y)+I(X;Y|X_r)+I(\hat{Y}_r;X,Y|X_r) \nonumber \\
&= I(X,X_r;Y)+I(\hat{Y}_r;X,Y|X_r).
\end{align}

Combining the bounds (\ref{rate_oneway_1}) and (\ref{rate_oneway_4}), we have
\begin{align}
\label{rate_oneway_5}
R&\leq I(X,X_r;Y)+I(\hat{Y}_r;X,Y|X_r)-I(\hat{Y}_r;Y_r|X_r) \nonumber \\
&=I(X,X_r;Y)+I(\hat{Y}_r;X,Y|X_r)-I(\hat{Y}_r;Y_r,X,Y|X_r) \nonumber \\
&=I(X,X_r;Y)-I(\hat{Y}_r;Y_r|X,X_r,Y).
\end{align}
Combining (\ref{rate_oneway_1}), (\ref{rate_oneway_2}) and (\ref{rate_oneway_3}), (\ref{rate_oneway_5}), we obtain the result of Theorem \ref{thm_rate}.
\end{proof}

\subsection{Comparison with the Original Compress-Forward Scheme}
\begin{thm}
\label{thm:oneway_ori}
Compress-forward without binning in Theorem \ref{thm_rate} achieves the same rate as the original compress-forward scheme for the one-way relay channel, which is:
\begin{align}
R\leq \min \{&I(X,X_r;Y)-I(\hat{Y}_r;Y_r|X,X_r,Y),\nonumber\\
&I(X;Y,\hat{Y}_r|X_r)\}
\end{align}
for some $p(x)p(x_r)p(\hat{y}_r|y_r,x_r)p(y,y_r|x,x_r)$.
\end{thm}

\begin{proof}
To show that the rate region in Theorem \ref{thm_rate} is the same as the rate region in Theorem \ref{thm:oneway_ori}, we need to show that the constraint (\ref{oneway_const_1}) is redundant. Note that an equivalent characterization of the rate region in Theorem \ref{thm:oneway_ori} is as follows \cite{cover1979capacity} \cite{gamal2006bounds} \cite{gamal2010lecture}:
\begin{align}
R\leq I(X;Y,\hat{Y}_r|X_r)
\end{align}
subject to
\begin{align}
\label{oneway_const_2}
I(X_r;Y)\geq I(\hat{Y}_r;Y_r|X_r,Y).
\end{align}
for some $p(x)p(x_r)p(\hat{y}_r|y_r,x_r)$.
Therefore, comparing (\ref{oneway_const_1}) with (\ref{oneway_const_2}), we only need to show that
\begin{align}
I(\hat{Y}_r;Y_r|X_r,Y)\geq I(\hat{Y}_r;Y_r|X_r)-I(\hat{Y}_r;X,Y|X_r).
\end{align}
This is true since
\begin{align}
I(\hat{Y}_r;Y_r|X_r,Y)&=I(\hat{Y}_r;Y_r,X|X_r,Y) \nonumber \\
&=I(X;\hat{Y}_r|X_r,Y)+I(Y_r;\hat{Y}_r|X,X_r,Y) \nonumber \\
&\geq I(Y_r;\hat{Y}_r|X,X_r,Y) \nonumber \\
&=I(\hat{Y}_r;X,Y,Y_r|X_r)-I(\hat{Y}_r;X,Y|X_r) \nonumber \\
&=I(\hat{Y}_r;Y_r|X_r)-I(\hat{Y}_r;X,Y|X_r). \nonumber
\end{align}
\end{proof}

\begin{rem}
If using successive decoding, the rate achieved by compress-forward without binning is strictly less than that with binning. Thus joint decoding is crucial for compress-forward without binning.
\end{rem}

\begin{rem}
Joint decoding does not help improve the rate of the original compress-forward with binning.
\end{rem}

\begin{rem}
The binning technique plays a role of allowing successive decoding instead of joint decoding, thus reduces decoding complexity. However, it has no impact on achievable rate for the one-way relay channel. This effect on decoding complexity is similar to that in decode-forward relaying, in which binning allows successive decoding \cite{cover1979capacity} while no binning requires backward decoding \cite{willems1985the}.
\end{rem}

\section{Two-way Relay Channel}\label{sec:two-way}
In this section, we extend compress-forward without Wyner-Ziv binning but with joint decoding of both the message and compression index to the two-way relay channel. Compared with the original compress-forward scheme for the two-way relay channel \cite{rankov2006achievable}, compress-forward without binning achieves a strictly larger rate region when the channel is asymmetric for two users. Compared with noisy network coding, it achieves a similar rate region but with an extra constraint on the compression rate. Under certain conditions, this constraint is redundant for the Gaussian TWRC. For such cases, compress-forward without binning achieves the same rate region as noisy network coding but has much less decoding delay.

\subsection{Achievable Rate Region for Compress-Forward without Binning}
\begin{thm}
\label{rate_two_way_nobin}
The following rate region is achievable for the two-way relay channel by using compress-forward without binning but with joint decoding:
\begin{align}
\label{rate_thm1}
R_1 \leq \min\{&I(X_1;Y_2,\hat{Y}_r|X_2,X_r),\\
&I(X_1,X_r;Y_2|X_2)-I(\hat{Y}_r;Y_r|X_1,X_2,X_r,Y_2)\} \nonumber \\
R_2 \leq \min\{&I(X_2;Y_1,\hat{Y}_r|X_1,X_r),\nonumber \\
&I(X_2,X_r;Y_1|X_1)-I(\hat{Y}_r;Y_r|X_1,X_2,X_r,Y_1)\}\nonumber
\end{align}
subject to
\begin{align}
\label{const_1}
I(\hat{Y}_r;Y_r|X_1,X_2,X_r,Y_1)&\leq I(X_r;Y_1|X_1)\nonumber\\
I(\hat{Y}_r;Y_r|X_1,X_2,X_r,Y_2)&\leq I(X_r;Y_2|X_2)
\end{align}
for some $p(x_1)p(x_2)p(x_r)p(y_1,y_2,y_r|x_1,x_2,x_r)p(\hat{y}_r|x_r,y_r)$.
\end{thm}

\begin{proof}
We use a block coding scheme in which each user sends $b-1$ messages over $b$ blocks of $n$ symbols each.\par
\subsubsection{Codebook generation}
 Fix code distribution $p(x_1)p(x_2)p(x_r)p(\hat{y}_r|x_r,y_r)$. We randomly and independently generate a codebook for each block $j\in [1:b]$
\begin{itemize}
\item Generate $2^{nR_1}$ i.i.d. sequences $x_1^n(m_{1,j})\sim\prod ^n_{i=1}p(x_{1i}) $, where $m_{1,j} \in [1:2^{nR_1}]$.
\item Generate $2^{nR_2}$ i.i.d. sequences $x_2^n(m_{2,j})\sim\prod ^n_{i=1}p(x_{2i}) $, where $m_{2,j} \in [1:2^{nR_2}]$.
\item Generate $2^{nR_r}$ i.i.d. sequences $x_r^n(k_{j-1})\sim\prod ^n_{i=1}p(x_{ri})$, where $k_{j-1} \in [1:2^{nR_r}]$.
\item For each $k_{j-1}\in [1:2^{nR_r}]$, generate $2^{nR_r}$ i.i.d. sequences $\hat{y}^n_r(k_j|k_{j-1})\sim\prod ^n_{i=1}p(\hat{y}_{ri}|x_{ri}(k_{j-1}))$, where $k_{j}\in [1:2^{nR_r}]$.
\end{itemize}

\subsubsection{Encoding}
 User 1 and user 2 transmits $x_1^n(m_{1,j})$ and $x_2^n(m_{2,j})$ in block $j$ respectively. The relay, upon receiving $y^n_r(j)$, finds an index $k_j$ such that $
((\hat{y}_r^n(k_j|k_{j-1}),y^n_r(j),x^n_r(k_{j-1}))\in A^n_{\epsilon '} $. Assume that such $k_j$ is found, the relay sends $x^n_r(k_j)$ in block $j+1$.\par

\subsubsection{Decoding}: We discuss the decoding at user 1. Assume that user 1 has decoded $k_{j-1}$ correctly in block $j$. Then in block $j+1$, user 1 finds a unique pair of $(\hat{m}_{2,j},\hat{k}_j)$ such that
\begin{align}
(x_2^n(\hat{m}_{2,j}),x_r^n(\hat{k}_{j-1}),\hat{y}_r^n(\hat{k}_j|\hat{k}_{j-1}),y_1^n(j),x_1^n(m_{1,j}))&\in A^n_{\epsilon} \nonumber\\
\textrm{and}~~~~~~~~~~~~~~(x_r^n(\hat{k}_j),y_1^n(j+1),x_1^n(m_{1,j+1}))&\in A^n_{\epsilon}. \nonumber
\end{align}

\subsubsection{Error analysis}
Assume without loss of generality that $m_{1,j}=m_{1,j+1}=m_{2,j}=1$ and $k_{j-1}=k_{j}=1$. First define the following two events:
\begin{align}
\mathcal{E}'_{1j}(k_j)=\big\{(x_r^n(k_j),y_1^n(j+1),x_1^n(1))\,&\in A^n_{\epsilon}\big\} \nonumber \\
\mathcal{E}'_{2j}(m_{2,j},k_j)=\big\{(x_2^n({m}_{2,j}),x_r^n(1),\hat{y}_r^n({k}_j|1)&, \nonumber \\
y_1^n(j),x_1^n(1))&\in A^n_{\epsilon}\big\}. \nonumber
\end{align}
Then the decoder makes an error only if one or more of the following events occur:
\begin{align}
\mathcal{E}_{1j}=&\big\{(\hat{y}_r^n(k_j|1),y^n_r(j),x^n_r(1))\notin A^n_{\epsilon '}~\textrm{for all}~k_j\in[1:2^{nR_r}]\big\} \nonumber \\
\mathcal{E}_{2j}=&\big\{(x_r^n(1),y_1^n(j+1),x_1^n(1))\notin A^n_{\epsilon}~\textrm{or}~\nonumber \\
&~~(x_2^n(1),x_r^n(1),\hat{y}_r^n(1|1),y_1^n(j),x_1^n(1))\notin A^n_{\epsilon} \big\} \nonumber\\
\mathcal{E}_{3j}=&\big\{\mathcal{E}'_{1j}(k_j)~\textrm{and}~\mathcal{E}'_{2j}(1,k_j)~\textrm{for some}~k_j\neq 1\big\} \nonumber \\
\mathcal{E}_{4j}=&\big\{\mathcal{E}'_{1j}(1)~\textrm{and}~\mathcal{E}'_{2j}(m_{2,j},1)~\textrm{for some}~m_{2,j}\neq 1\big\} \nonumber \\
\mathcal{E}_{5j}=&\big\{\mathcal{E}'_{1j}(k_j)~\textrm{and}~\mathcal{E}'_{2j}(m_{2,j},k_j)~\textrm{for some}~m_{2,j}\neq 1,k_j\neq 1\big\}. \nonumber
\end{align}
Thus, the probability of error is bound as
\begin{align}
P\{\hat{m}_{2,j}\neq 1,\hat{k_j}\neq 1\}\leq & P(\mathcal{E}_{1j})+P(\mathcal{E}_{2j}\cap\mathcal{E}_{1j}^c)+P(\mathcal{E}_{3j})\nonumber \\
&+P(\mathcal{E}_{4j})+P(\mathcal{E}_{5j}). \nonumber
\end{align}
By the covering lemma, $P(\mathcal{E}_{1j}) \rightarrow 0$ as $n\rightarrow \infty$, if
\begin{align}
\label{rate_1}
R_r > I(\hat{Y}_r;Y_r|X_r).
\end{align}
By the conditional typicality lemma, $P(\mathcal{E}_{2j}\cap\mathcal{E}_{1j}^c) \rightarrow 0$ as $n\rightarrow \infty$. \par
For the rest of the error events, the decoded joint distribution for each event is as follows.
\begin{align*}
\mathcal{E}'_{1j}(k_j)&: p(x_1)p(x_r)p(y_1|x_1)\\
\mathcal{E}'_{2j}(1,k_j)&: p(x_1)p(x_2)p(x_r)p(\hat{y}_r|x_r)p(y_1|x_2,x_r,x_1)\\
\mathcal{E}'_{2j}(m_{2,j},1)&: p(x_1)p(x_2)p(x_r)p(y_1,\hat{y}_r|x_r,x_1)\\
\mathcal{E}'_{2j}(m_{2,j},k_j)&: p(x_1)p(x_2)p(x_r)p(\hat{y}_r|x_r)p(y_1|x_r,x_1),
\end{align*}
where $m_{2,j} \neq 1, k_j \neq 1$. Using standard joint typicality analysis with the above decoded joint distribution, we can obtain a bound on each error event as follows.
\begin{align}
P(\mathcal{E}_{3j})&\leq 2^{nR_r}\cdot 2^{-n(I(X_r;Y_1|X_1)-\delta(\epsilon))}\nonumber \\
 &~~~\cdot 2^{-n(I(\hat{Y}_r;X_1,X_2,Y_1|X_r)-4\delta(\epsilon))} \nonumber\\
P(\mathcal{E}_{4j})&\leq 2^{nR_2}\cdot 2^{-n(I(X_2;Y_1,\hat{Y}_r|X_r,X_1)-3\delta(\epsilon))} \nonumber\\
P(\mathcal{E}_{5j})&\leq 2^{nR_2}\cdot 2^{nR_r}\cdot 2^{-n(I(X_r;Y_1|X_1)-\delta(\epsilon))}\nonumber\\
 &~~~\cdot 2^{-n(I(X_2;Y_1|X_1,X_r)+I(\hat{Y}_r;X_1,X_2,Y_1|X_r)-4\delta(\epsilon))} \nonumber
\end{align}
All of them tend to zero as $n\rightarrow \infty $ if
\begin{align}
\label{rate_Rr}
R_r &\leq I(X_r;Y_1|X_1)+I(\hat{Y}_r;X_1,X_2,Y_1|X_r)\\
\label{rate_twoway_5}
R_2 &\leq I(X_2;Y_1,\hat{Y}_r|X_r,X_1)\\
\label{rate_4}
R_2+R_r &\leq I(X_r;Y_1|X_1)+I(X_2;Y_1|X_1,X_r)\\
&~~~+I(\hat{Y}_r;X_1,X_2,Y_1|X_r) \nonumber \\
&=I(X_2,X_r;Y_1|X_1)+I(\hat{Y}_r;X_1,X_2,Y_1|X_r).\nonumber
\end{align}

Combining the bounds (\ref{rate_1}) and (\ref{rate_4}), we have
\begin{align}
\label{rate_7}
R_2 \leq &I(X_2,X_r;Y_1|X_1)+I(\hat{Y}_r;X_1,X_2,Y_1|X_r)\nonumber\\
&-I(\hat{Y}_r;Y_r|X_r)\nonumber\\
=&I(X_2,X_r;Y_1|X_1)-I(\hat{Y}_r;Y_r|X_1,X_2,X_r,Y_1).
\end{align}
Combining the bounds (\ref{rate_twoway_5}) and (\ref{rate_7}), we obtain the rate constraint on $R_2$ in Theorem \ref{rate_two_way_nobin}. Similar for $R_1$. From (\ref{rate_1}) and (\ref{rate_Rr}), we obtain constraint (\ref{const_1}).
\end{proof}
Next we show the achievable rate region for the Gaussian two-way relay channel using compress-forward without binning. Recall the Gaussian two-way relay channel model (\ref{GTWRC}) in Section \ref{sec:system_model}. Assume $X_1 \sim \mathcal{N}(0,P), X_2 \sim \mathcal{N}(0,P), X_r \sim \mathcal{N}(0,P), \hat{Z} \sim \mathcal{N}(0,\sigma^2)$ to be independent, and $\hat{Y}_r=Y_r+\hat{Z}$. Denote
\begin{align}
\label{four_rate}
R_{11}(\sigma^2)&=C\left(g_{21}^2P+\frac{g_{r1}^2P}{1+\sigma^2}\right)\nonumber\\
R_{12}(\sigma^2)&=C(g_{21}^2P+g_{2r}^2P)-C(1/\sigma^2)\nonumber\\
R_{21}(\sigma^2)&=C\left(g_{12}^2P+\frac{g_{r2}^2P}{1+\sigma^2}\right)\nonumber\\
R_{22}(\sigma^2)&=C(g_{12}^2P+g_{1r}^2P)-C(1/\sigma^2).
\end{align}
Then we have the following rate regions for the Gaussian two-way relay channel.
\begin{cor}
\label{cor:G_nobin}
The following rate region is achievable for the Gaussian two-way relay channel using compress-forward without binning:
\begin{align}
\label{Gaussian_compress}
R_1&\leq\min\{R_{11}(\sigma^2),R_{12}(\sigma^2)\}\nonumber\\
R_2&\leq\min\{R_{21}(\sigma^2),R_{22}(\sigma^2)\}
\end{align}
for some $\sigma^2 \geq \max \{\sigma_{c1}^2,\sigma_{c2}^2\}$, where
\begin{align}
\label{c1c2}
\sigma_{c1}^2&=(1+g_{21}^2P)/(g_{2r}^2P)\nonumber\\
\sigma_{c2}^2&=(1+g_{12}^2P)/(g_{1r}^2P).
\end{align}
and $R_{11}(\sigma^2), R_{12}(\sigma^2), R_{21}(\sigma^2), R_{22}(\sigma^2)$ are defined in (\ref{four_rate}).
\end{cor}

\subsection{Comparison with the Original Compress-Forward Scheme}
In this section, we first present the rate region achieved by the original compress-forward scheme for the two way relay channel \cite{rankov2006achievable}. We then show that compress-forward without
Wyner-Ziv binning but with joint decoding can achieve a larger rate region.

\begin{thm}
\label{thm_rate_ori}
\em{[Rankov and Wittneben].}
\emph{
The following rate region is achievable for two-way relay channel with compress-forward scheme:
\begin{align}
\label{rate_thm2}
R_1& \leq I(X_1;Y_2,\hat{Y}_r|X_2,X_r) \nonumber \\
R_2& \leq I(X_2;Y_1,\hat{Y}_r|X_1,X_r)
\end{align}
subject to
\begin{align}
\label{const_2}
&\max \{I(\hat{Y}_r;Y_r|X_1,X_r,Y_1),I(\hat{Y}_r;Y_r|X_2,X_r,Y_2)\}\nonumber\\
&\leq \min \{I(X_r;Y_1|X_1),I(X_r;Y_2|X_2)\}
\end{align}
for some $p(x_1)p(x_2)p(x_r)p(y_1,y_2,y_r|x_1,x_2,x_r)p(\hat{y}_r|x_r,y_r)$.
}
\end{thm}

Next we present a short proof to show the difference from compress-forward without binning. The proof follows the same lines as in \cite{rankov2006achievable}, but we also correct an error in the analysis in \cite{rankov2006achievable} as point out in Remark \ref{rem:error}.
\begin{proof}
We use a block coding scheme in which each user sends $b-1$ messages over $b$ blocks of $n$ symbols each.\par
\subsubsection{Codebook generation}
 Fix code distribution $p(x_1)p(x_2)p(x_r)p(\hat{y}_r|x_r,y_r)$. We randomly and independently generate a codebook for each block $j\in [1:b]$
\begin{itemize}
\item Generate $2^{nR_1}$ i.i.d. sequences $x_1^n(m_{1,j})\sim\prod ^n_{i=1}p(x_{1i}) $, where $m_{1,j} \in [1:2^{nR_1}]$.
\item Generate $2^{nR_2}$ i.i.d. sequences $x_2^n(m_{2,j})\sim\prod ^n_{i=1}p(x_{2i}) $, where $m_{2,j} \in [1:2^{nR_2}]$.
\item Generate $2^{nR_r}$ i.i.d. sequences $x_r^n(q_{j-1})\sim\prod ^n_{i=1}p(x_{ri})$, where $q_{j-1} \in [1:2^{nR_r}]$.
\item For each $q_{j-1}\in [1:2^{nR_r}]$, generate $2^{n(R_r+R'_r)}$ i.i.d. sequences $\hat{y}^n_r(q_j,r_j|q_{j-1})\sim\prod ^n_{i=1}p(\hat{y}_{ri}|x_{ri}(q_{j-1}))$. Throw them into $2^{nR_r}$ bins, where $q_j\in [1:2^{nR_r}]$ denotes the bin index and $r_j\in [1:2^{nR'_r}]$ denotes the relative index within a bin.
\end{itemize}

\subsubsection{Encoding}
 User 1 and user 2 transmits $x_1^n(m_{1,j})$ and $x_2^n(m_{2,j})$ in block $j$ separately. The relay, upon receiving $y^n_r(j)$, finds an index pair $(q_j,r_j)$ such that $
((\hat{y}_r^n(q_j,r_j|q_{j-1}),y^n_r(j),x^n_r(q_{j-1}))\in A^n_{\epsilon '} $. Assume that such $(q_j,r_j)$ is found, the relay sends $x^n_r(q_j)$ in block $j+1$. By the covering lemma, the probability that there is no such $(q_j,r_j)$ tends to 0 as $n\rightarrow \infty$ if
\begin{align}
\label{rate_twoway_ori_1}
R_r+R'_r > I(\hat{Y}_r;Y_r|X_r).
\end{align}

\subsubsection{Decoding}
Each user applies 3-step successive decoding.
At the end of block $j$, user 1 determines the unique bin $\hat{q}_{j-1}$ such that
\begin{align*}
(x^n_r(\hat{q}_{j-1}),y^n_1(j),x^n_1(m_{1,j})) \in A^n_{\epsilon}.
\end{align*}
Similar for user 2. Both succeed with high probability if
\begin{align}
\label{rate_twoway_ori_2}
R_r\leq \min \{I(X_r;Y_1|X_1),I(X_r;Y_2|X_2)\}.
\end{align}
Then user 1 uses $y^n_1(j-1)$ to determine the unique $\hat{r}_{j-1}$ such that
\begin{align*}
(\hat{y}^n_r(\hat{q}_{j-1},\hat{r}_{j-1}|\hat{q}_{j-2}),y^n_1(j-1),x^n_r(\hat{q}_{j-2}),x^n_1(m_{1,j-1}))\in A^n_{\epsilon}.
\end{align*}
Similar for user 2. Both succeed with high probability if
\begin{align}
\label{rate_twoway_ori_3}
R'_r\leq \min \{I(\hat{Y}_r;X_1,Y_1|X_r),I(\hat{Y}_r;X_2,Y_2|X_r)\}.
\end{align}
Finally, user 1 uses both $y^n_1(j-1)$ and $\hat{y}^n_r(j-1)$ to determines the unique $\hat{m}_{2,j-1}$ such that
\begin{align*}
(x^n_2(\hat{m}_{2,j-1}),\hat{y}^n_r(\hat{q}_{j-1},\hat{r}_{j-1}|\hat{q}_{j-2}),y^n_1(j-1),x^n_r(\hat{q}_{j-2}), & \\x^n_1(m_{1,j-1}))\in A^n_{\epsilon}.&
\end{align*}
Similar for user 2. Both succeed with high probability if
\begin{align}
R_1& \leq I(X_1;Y_2,\hat{Y}_r|X_2,X_r)\nonumber\\
R_2& \leq I(X_2;Y_1,\hat{Y}_r|X_1,X_r).
\end{align}
The constraint (\ref{const_2}) comes from (\ref{rate_twoway_ori_1}), (\ref{rate_twoway_ori_2}) and (\ref{rate_twoway_ori_3}).
\end{proof}

\begin{rem}
\label{rem:error}
We note an error in the proof of \cite{rankov2006achievable}. In \cite{rankov2006achievable}, when user 1 determines the unique $\hat{r}_{j-1}$, it is stated that it succeeds with high probability if
\begin{align}
\label{wrong}
R'_r\leq I(\hat{Y}_r;Y_1|X_1,X_r),
\end{align}
which corresponds to (\ref{rate_twoway_ori_3}) in our analysis. However, this is incorrect since the decoded joint distribution of this error event is $p(x_1)p(x_r)p(\hat{y}_r|x_r)p(y_1|x_1,x_r)$. Therefore, the error probability can be bounded as
\begin{align*}
&P(\mathcal{E})=\sum _{(x_1,x_r,\hat{y}_r,y_1)\in A^n_{\epsilon}} p(x_1)p(x_r)p(\hat{y}_r|x_r)p(y_1|x_1,x_r)\\
&\leq 2^{n(H(X_1,X_r,\hat{Y}_r,Y_1)-H(X_1)-H(X_r)-H(\hat{Y}_r|X_r)-H(Y_1|X_1,X_r)-3\delta(\epsilon))}\\
&=2^{-n(I(\hat{Y}_r;X_1,Y_1|X_r)-3\delta(\epsilon))},
\end{align*}
which tends to zero as as $n\rightarrow \infty $ if (\ref{rate_twoway_ori_3}) is satisfied instead of (\ref{wrong}).
\end{rem}

\begin{figure}[t]
\centering
  \includegraphics[scale=0.6]{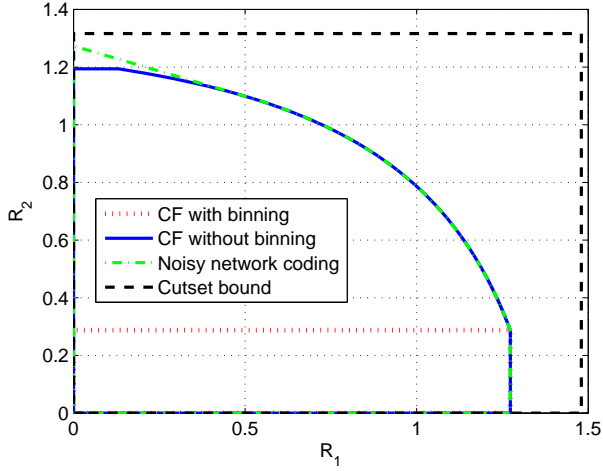}\\
  \caption{Rate regions for
    $P=20,g_{r1}=g_{1r}=2,g_{r2}=g_{2r}=0.5,g_{12}=g_{21}=0.1$}\label{s1}
\end{figure}

\begin{cor}
In the two-way relay channel, the rate region achieved by compress-forward without binning in Theorem \ref{rate_two_way_nobin} is larger than the rate region achieved by the original compress-forward scheme in Theorem
\ref{thm_rate_ori} when the channel is asymmetric for the two users. The two regions are equal if and only if the following conditions holds:
\begin{align}
I(X_r;Y_1|X_1)&=I(X_r;Y_2|X_2)\nonumber \\
I(\hat{Y}_r;Y_r|X_1,X_r,Y_1)&=I(\hat{Y}_r;Y_r|X_2,X_r,Y_2).
\label{eq_con}
\end{align}
\end{cor}

 \begin{proof}
 First, we show that the constraint of Theorem \ref{rate_two_way_nobin} is looser than that of Theorem \ref{thm_rate_ori}. This is true since from (\ref{const_2}) we have
\begin{align}
I(X_r;Y_1|X_1)& \geq I(\hat{Y}_r;Y_r|X_1,X_r,Y_1)\nonumber\\
&=I(\hat{Y}_r;X_2,Y_r|X_1,X_r,Y_1)\nonumber\\
&\geq I(\hat{Y}_r;Y_r|X_1,X_2,X_r,Y_1)
\label{rate_13}
\end{align}
where (\ref{rate_13}) is the right hand side of the first term in (\ref{const_1}). Similar for the other term.\par
Next we show that (\ref{rate_thm2}) and (\ref{const_2}) imply (\ref{rate_thm1}). From (\ref{rate_thm2}) we have
\begin{align}
R_2\leq &I(X_2;Y_1,\hat{Y}_r|X_1,X_r)\nonumber \\
=&I(X_2;Y_1|X_1,X_r)+I(\hat{Y}_r;X_2|Y_1,X_1,X_r)\nonumber \\
=&I(X_2,X_r;Y_1|X_1)-I(X_r;Y_1|X_1)+I(\hat{Y}_r;X_2|Y_1,X_1,X_r)\nonumber \\
\overset{(a)}{\leq} &I(X_2,X_r;Y_1|X_1)-I(\hat{Y}_r;Y_r|X_1,X_r,Y_1)\nonumber \\
&+I(\hat{Y}_r;X_2|Y_1,X_1,X_r)\nonumber \\
=&I(X_2,X_r;Y_1|X_1)-I(\hat{Y}_r;Y_r|X_1,X_2,X_r,Y_1)\nonumber
\end{align}
where $(a)$ follows from the constraint of (\ref{const_2}) in Theorem \ref{thm_rate_ori}. The equality holds when
\begin{align}
&I(X_r;Y_1|X_1)=\min\{I(X_r;Y_1|X_1),I(X_r;Y_2|X_2)\}\nonumber \\
&I(\hat{Y}_r;Y_r|X_1,X_r,Y_1)\nonumber\\
&=\max (I(\hat{Y}_r;Y_r|X_1,X_r,Y_1),I(\hat{Y}_r;Y_r|X_2,X_r,Y_2)).\nonumber
\end{align}
Similar for $R_1$, the equality holds when
\begin{align}
&I(X_r;Y_2|X_2)=\min\{I(X_r;Y_1|X_1),I(X_r;Y_2|X_2)\}\nonumber \\
&I(\hat{Y}_r;Y_r|X_2,X_r,Y_2)\nonumber\\
&=\max (I(\hat{Y}_r;Y_r|X_1,X_r,Y_1),I(\hat{Y}_r;Y_r|X_2,X_r,Y_2)).\nonumber
\end{align}
Therefore, if and only if condition (\ref{eq_con}) holds,
the original compress-forward scheme achieves same rate region as compress-forward scheme without binning; otherwise, it will be strictly smaller.
\end{proof}

\begin{rem}
For the two-way relay channel, compress-forward without binning achieves larger rate region than the original compress-forward scheme when the channel is asymmetric for the two users. This is because, with binning and successive decoding, the compression rate is limited by the weaker one of the channels from the relay to two users. But without binning, this limitation is removed.
\end{rem}

\begin{figure}[t]
\centering
  \includegraphics[scale=0.6]{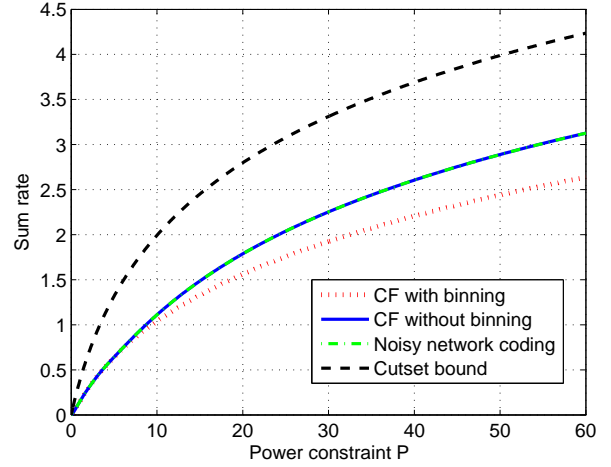}\\
  \caption{Sum rate for
    $g_{r1}=g_{1r}=2,g_{r2}=g_{2r}=0.5,g_{12}=g_{21}=0.1$}\label{s2}
\end{figure}

\subsection{Comparison with Noisy Network Coding}
In this section, we compare the rate region achieved by compress-forward without binning with that achieved by noisy network coding \cite{sung2011noisy} for the two way relay channel.
\begin{thm}
\label{thm_noisy}
\em{[Lim, Kim, El Gamal, and Chung].} \emph{The following rate region is achievable for the two-way relay channel using noisy network coding:
\begin{align}
\label{rate_noisy}
R_1 \leq \min\{&I(X_1;Y_2,\hat{Y}_r|X_2,X_r),\\
&I(X_1,X_r;Y_2|X_2)-I(\hat{Y}_r;Y_r|X_1,X_2,X_r,Y_2)\} \nonumber \\
R_2 \leq \min\{&I(X_2;Y_1,\hat{Y}_r|X_1,X_r),\nonumber \\
&I(X_2,X_r;Y_1|X_1)-I(\hat{Y}_r;Y_r|X_1,X_2,X_r,Y_1)\}\nonumber
\end{align}
for some $p(x_1)p(x_2)p(x_r)p(y_1,y_2,y_r|x_1,x_2,x_r)p(\hat{y}_r|x_r,y_r)$.
}
\end{thm}
Comparing Theorem \ref{rate_two_way_nobin} with Theorem \ref{thm_noisy}, we find that the rate constraints for $R_1$ and $R_2$ in compress-forward without binning (\ref{rate_thm1}) are the same as those in noisy network coding (\ref{rate_noisy}). However, compress-forward without binning has an extra constraint on the compression rate as in (\ref{const_1}). Therefore, in general, noisy network coding achieves a larger rate region than compress-forward without binning. Next, we show that for the Gaussian two-way relay channel, these two schemes achieve same region under certain conditions.\par

\begin{figure}[t]
\centering
  \includegraphics[scale=0.6]{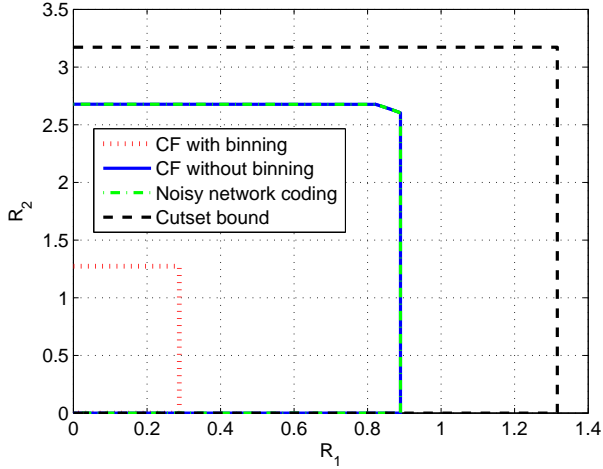}\\
  \caption{Rate regions for
    $P=20,g_{r1}=0.5,g_{1r}=2,g_{r2}=2,g_{2r}=0.5,g_{12}=g_{21}=0.1$}\label{s5}
\end{figure}

\begin{figure}[t]
\centering
  \includegraphics[scale=0.6]{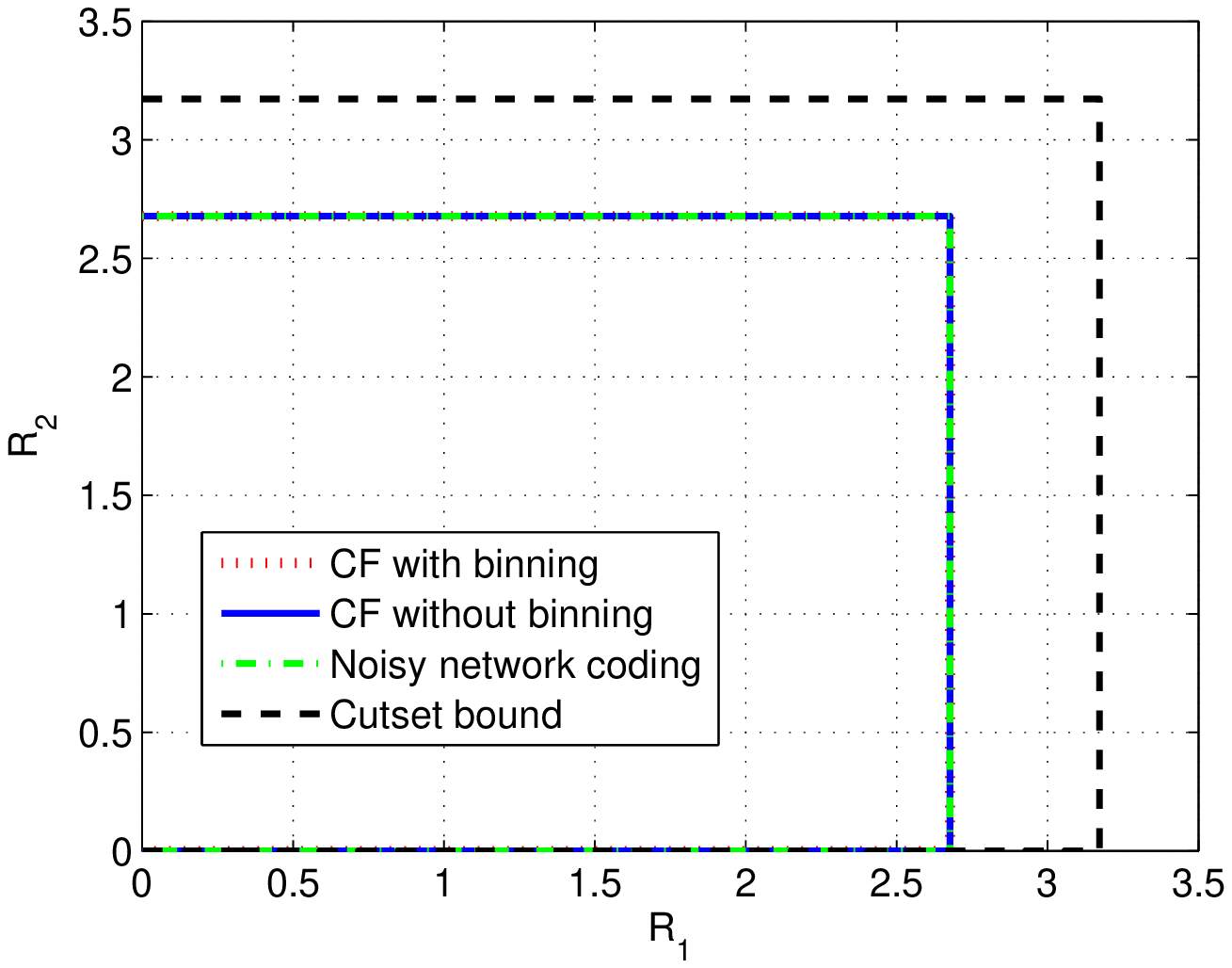}\\
  \caption{Rate regions for
    $P=20,g_{r1}=g_{1r}=g_{r2}=g_{2r}=2,g_{12}=g_{21}=0.1$}\label{s3}
\end{figure}

\begin{cor}
\em{[Lim, Kim, El Gamal, and Chung].}
\emph{The following rate region is achievable for the Gaussian two-way relay channel with noisy network coding scheme:
\begin{align}
\label{Gaussian_noisy}
R_1&\leq\min\{R_{11}(\sigma^2),R_{12}(\sigma^2)\}\nonumber\\
R_2&\leq\min\{R_{21}(\sigma^2),R_{22}(\sigma^2)\}
\end{align}
for some $\sigma^2>0$, where $R_{11}(\sigma^2), R_{12}(\sigma^2), R_{21}(\sigma^2), R_{22}(\sigma^2)$ are defined in (\ref{four_rate}).}
\end{cor}

\begin{thm}
Compress-forward without binning achieves the same rate region as noisy network coding for the Gaussian two-way relay channel if
\begin{align}
\label{condition_2}
\sigma_{c1}^2&\leq\sigma_{e2}^2\nonumber\\
\sigma_{c2}^2&\leq\sigma_{e1}^2,
\end{align}
where
\begin{align}
\sigma_{e1}^2&=(1+g_{21}^2P+g_{r1}^2P)/(g_{2r}^2P)\nonumber\\
\sigma_{e2}^2&=(1+g_{12}^2P+g_{r2}^2P)/(g_{1r}^2P),\nonumber
\end{align}
and $\sigma_{c1}^2,\sigma_{c2}^2$ are defined in (\ref{c1c2}). Otherwise it achieves a smaller rate region.
\end{thm}

\begin{proof}
Note that both $R_{11}(\sigma^2),R_{21}(\sigma^2)$ are nonincreasing and $R_{12}(\sigma^2),R_{22}(\sigma^2)$ are nondecreasing. Also,
\begin{align}
R_{11}(\sigma_{e1}^2)&=R_{12}(\sigma_{e1}^2)\nonumber\\
R_{21}(\sigma_{e2}^2)&=R_{22}(\sigma_{e2}^2).\nonumber
\end{align}
Therefore, the constraint in Corollary \ref{cor:G_nobin} is redundant if
\begin{align}
\label{condition_1}
\max \{\sigma_{c1}^2,\sigma_{c2}^2\}\leq \min \{\sigma_{e1}^2,\sigma_{e2}^2\}.
\end{align}
Since $\sigma_{c1}^2\leq\sigma_{e1}^2$ and $\sigma_{c2}^2\leq\sigma_{e2}^2$ always hold, the above condition (\ref{condition_1}) is equivalent to (\ref{condition_2}).
\end{proof}



\subsection{Numerical Results}
We now compare numerically compress-forward without binning with the original compress-forward scheme \cite{rankov2006achievable} and noisy network coding \cite{sung2011noisy}. Figure \ref{s1} shows an asymmetric channel configuration in which compress-forward without binning achieves a strictly larger rate region than the original compress-forward scheme, but slightly smaller than noisy network coding. Figure \ref{s2} plots the sum rates for the same channel configurations as in Figure \ref{s1}, which shows that compress-forward without binning achieves the same sum rate as noisy network coding despite a smaller rate region. Figure \ref{s5} shows a case where compress-forward without binning achieves the same rate region as noisy network coding and larger than compress-forward with binning. Figure \ref{s3} shows a symmetric channel configuration in which three schemes achieve the same rate region. Figure \ref{s4} shows the sum rate for another configuration when the two users and the relay are on a straight line. In this case, compress-forward without binning and noisy network coding achieves the same sum rate, which is higher than that of compress-forward with binning.\par
As shown in these figures, compress-forward without binning achieves larger rate region and sum rate than the original compress-forward scheme in \cite{rankov2006achievable} when the channel is asymmetric for the two users. Compress-forward without binning achieves the same rate region as noisy network coding incases when (\ref{condition_2}) is satisfied. Furthermore, it has less decoding delay which is only 1 instead of $b$ blocks.

\begin{figure}[t]
\centering
  \includegraphics[scale=0.6]{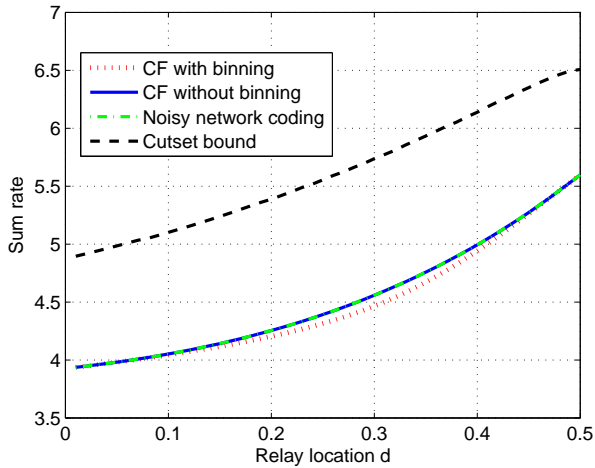}\\
  \caption{Sum rate for
    $P=10, g_{r1}=g_{1r}=d^{-\gamma/2}, g_{r2}=g_{2r}=(1-d)^{-\gamma/2},g_{12}=g_{21}=1, \gamma=3$}\label{s4}
\end{figure}

\section{Conclusion}\label{Con}
We have analyzed compress-forward without Wyner-Ziv binning but with joint
decoding of both the message and compression index in the one-way and two-way relay channels.
In both channels, compress-forward without binning either achieves the same rate or improves the rate region compared to the original compress-forward with binning by not limiting the compression rate to the weakest link. It does however increase decoding complexity by requiring joint decoding instead of sequential decoding. Compress-forward without binning also achieves the same rate regions as noisy network coding for the Gaussian TWRC in certain cases, for which it may be more preferable because of less decoding delay. These results help understand the role of Wyner-Ziv binning in compress-forward
schemes.

\section*{Acknowledgment}
The authors would like to thank Ahmad Abu Al Haija for checking a part of the proof.

\bibliographystyle{IEEEtran}
\bibliography{reflist}

\end{document}